\documentclass[journal]{IEEEtran}

\usepackage{moreverb}
\usepackage{epsfig}
\usepackage{amsmath,amssymb,amsthm,mathrsfs,amsfonts,dsfont}
\usepackage{adjustbox,lipsum}
\usepackage[linesnumbered,ruled,vlined]{algorithm2e}
\usepackage{amsfonts}
\usepackage{epsfig}
\usepackage{amssymb}
\usepackage{amsmath}
\usepackage{amsthm}
\usepackage{multirow}
\usepackage{setspace}
\usepackage{rotating}
\usepackage{graphicx}
\usepackage{tabularx}
\usepackage{array}
\usepackage{anyfontsize}
\usepackage{color,soul}
\usepackage{graphicx} 
\usepackage{epstopdf}
\usepackage{blindtext}
\usepackage{amsmath}
\usepackage{amsthm,amssymb,amsmath,bm}
\usepackage{subfigure}
\usepackage{amsfonts}
\usepackage{epsfig}
\usepackage{amssymb}
\usepackage{amsmath}
\usepackage{cite}
\usepackage{graphicx}
\usepackage{fancyhdr}
\usepackage{caption}
\usepackage{tabularx}
\usepackage{tcolorbox}
\usepackage{cite}
\usepackage{setspace}

\usepackage[inline]{enumitem}

\usepackage{texdef2020}

\newtheorem{theorem}{Theorem}
\newtheorem{lemma}{Lemma}

\newtheorem{definition}{Definition}
\newtheorem{remark}{Remark}

\newcommand{\instant}{\tau}

\allowdisplaybreaks
\title{Age Dispersion and Higher-Order AoI
in Status Update Systems}

\author{ 
\IEEEauthorblockN{Mohammad~Moltafet, Roy~D.~Yates, Marian~Codreanu, and   Hamid~R.~Sadjadpour}

\thanks{
M. Moltafet is with the Department
of Electrical and Computer Engineering, Tennessee Tech University, TN, USA
 (e-mail: mmoltafet@tntech.edu).  R. D. Yates is with the Department of Electrical and Computer Engineering, Rutgers University, NJ, USA (e-mail: ryates@winlab.rutgers.edu). M. Codreanu is with the Department of Science and Technology, Link\"oping University, Link\"oping, Sweden (e-mail: marian.codreanu@liu.se). H. R. Sadjadpour is with the Department of Electrical and Computer Engineering, University of California, CA, USA (e-mail: hamid@ucsc.edu).
}
}
  
\begin{document}
\maketitle
\sloppy

\begin{abstract}
We introduce and characterize \emph{age dispersion} as a measure of temporal consistency in status update systems. Age dispersion is defined as the difference between the ages of the two most recently received updates, and its higher-order extension, the {$k$-th} order age dispersion,  captures the difference between the ages of the most recent update and the $(k+1)$-th most recent one. We analyze age dispersion in an M/G/1/1 queueing system. Furthermore, we establish connections between the {$k$-th} order age dispersion and the {$k$-th} order age of information (AoI), where the latter quantifies the age of the {$k$-th} most recently received update. 

\end{abstract}


 \section{Introduction}
The Age of Information (AoI) was introduced in \cite{5984917} as a metric to quantify the freshness of information in status update systems. A status update packet contains the measured value of a monitored process (i.e., a sample), along with a timestamp indicating the sample generation time. When an update has timestamp $u$ at time $\instant$, then its age is $\instant-u$. In addition,  we say that one update is fresher than another if its age is less.  At any time instant $\instant$, if the freshest received update has timestamp $U(\instant)$, then AoI is defined as the random process 
$\Delta(\instant) = \instant - U(\instant)$ \cite{5984917}.

In recent years, AoI has been examined from diverse perspectives within status update systems, including queueing models analysis, e.g., in \cite{6195689,8469047, 9013935,9048914,9162681,Moltafet2020mgf,9705518,9119460, 8820073,9099557,8886357,10899900,11015332,moltafetisit2025}, and AoI-optimal status updating control procedures, e.g., in \cite{8943134,9181539,8648525,9540757,MolACM2023,zakeri2023minimizing,ssvtvt,10619121}.  Comprehensive literature reviews of recent works in AoI from different perspectives can be found in \cite{9380899,10542348,Abbas_AoI_Sur,10286022,10893697}.

While AoI effectively captures the freshness of information, it does not fully reflect the temporal consistency of updates \cite{Liu_concurrency}. In some applications, status updates need to arrive not only fresh but also closely spaced in time. To address this, we define a new metric, the \textit{age dispersion}
at the sink as the absolute difference between the ages of the two freshest received updates.

As a motivating example, consider the \textit{remote industrial support using cloud-rendered augmented reality (AR)} application. A technician at a factory or power plant faces a complex machine requiring inspection or repair and receives real-time assistance from a remote expert. The technician wears AR glasses equipped with various sensors, such as a camera, eye tracker, and microphone, which collect data including head pose, camera feed, eye gaze, and audio. This data is sent to a cloud server for processing and rendering overlays. The processed data is then transmitted to the remote expert, who views the technician’s environment in real time via a monitor or VR headset. Annotations and 3D instructions generated in the cloud are streamed back to the technician’s AR display, where they appear as visual overlays, such as arrows pointing to screws, animations showing parts being removed, and floating text instructions. In this example:
    \textit{freshness (low AoI)} is critical to avoid misaligned overlays; otherwise, arrows, labels, or instructions may appear in incorrect locations, leading to confusion or errors in action, but
    \textit{low age dispersion} is equally important to maintain realism and immersion, and to avoid motion sickness by ensuring visual stability.

In addition to the age dispersion, we define the {\em $k$-th order age dispersion} as the difference between the age of the freshest received update and that of the $(k+1)$-th freshest update.
Moreover, we define the {\em $k$-th order AoI} as the age of the $k$-th freshest update.  We further show that the average $k$-th order AoI can be represented as a function of the AoI and  $k$-th order age dispersion. 

We characterize the average $k$-th order AoI in a single-source M/G/1/1 queueing system under a probabilistically preemptive policy \cite{moltafetisit2025}. According to the probabilistically preemptive policy, if the server is idle, an arriving packet immediately enters service; if the server is busy, the arriving packet preempts the one in service with probability $\theta$, and is otherwise discarded.

\section{ Age Dispersion Definition}\label{System Model}
The source monitors a random process, and the sink is interested in receiving updates that are both \textit{timely} and \textit{temporally close}. We consider a status update system consisting of one source and one sink, as depicted in Fig.~\ref{Model}. In this system, packets are received at the sink in order of increasing time-stamps. Under this assumption, we present a formal definition of the age dispersion.

\begin{figure}
\centering
\includegraphics[width=.8\linewidth,trim = 0mm 0mm 0mm 0mm,clip]{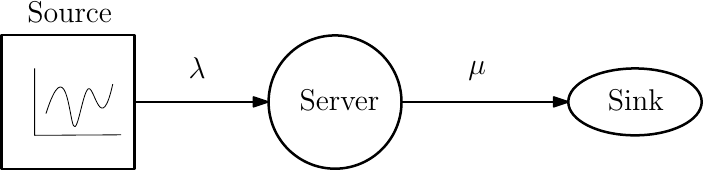}
\caption{The considered status update system.}
\label{Model}
\vspace{-5mm}
\end{figure}

\subsection{Age Dispersion}
With packets delivered in increasing time-stamp order, the age dispersion at the sink can be defined as the difference between the ages of the two most recently received updates.  In other words, the age dispersion is the difference between the AoI, i.e., age of the most recently received update, and the age of the previous update. 

\begin{figure}[t]
\centering
\includegraphics[width=1\linewidth,trim = 0mm 0mm 0mm 0mm,clip]{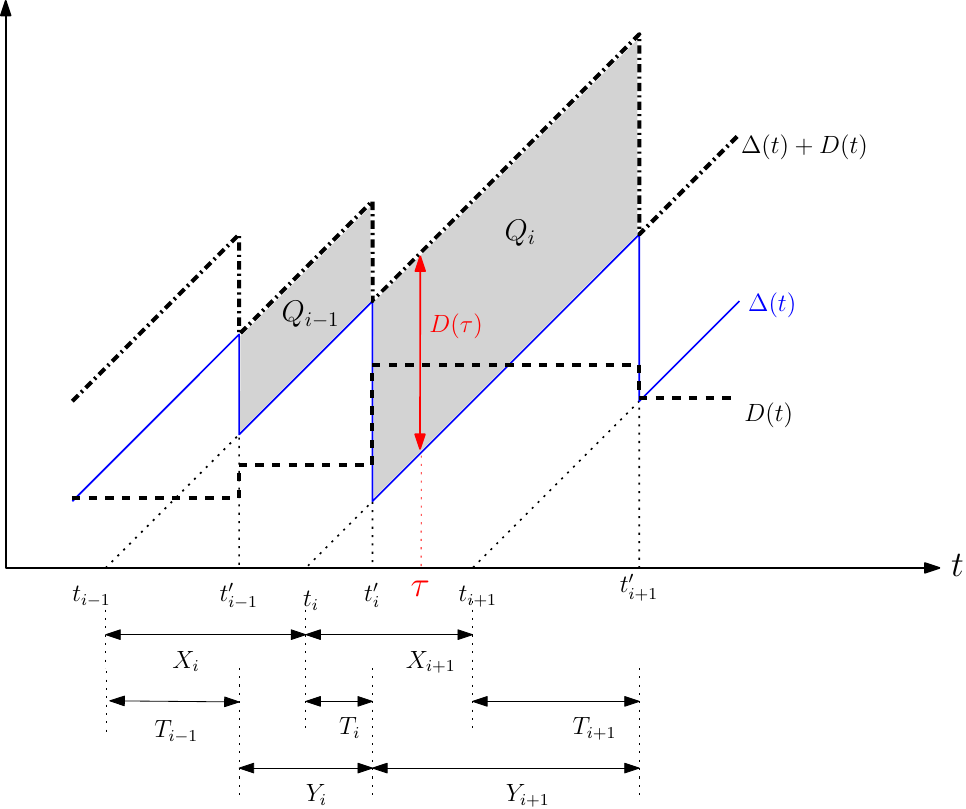}
\caption{An example of the AoI evolution. 
}
\label{Age dispersion}
\vspace{-5mm}
\end{figure}
Let $t'_i$ denote the time instant when the $i$th  update packet is delivered to the sink. In addition, let $t_i$ denote the time instant at which this same update is generated. 
At a time instant $\tau$, the index of the most recently received status update packet is
\begin{equation}\label{mnb00}
N(\tau)=\max\{i\colon t'_{i}\le \tau\}, 
\end{equation}
and the time stamp of the most recently received status update packet is
$t_{N(\tau)}$.
 The age dispersion at the time instant 
$\instant$ at the sink is defined as
\begin{align}
D(\instant)=t_{N(\instant)}-t_{N(\instant)-1},\qquad \instant \ge t'_2.
\end{align}


A sample path of the AoI exemplified in Fig.~\ref{Age dispersion} also illustrates the stepwise evolution of the age dispersion.
Averaging over the interval $(t'_2,t'_n)$ from the second update delivery through the $n$th update delivery at time  $t'_n$, the time average age dispersion at the sink
is defined as
\begin{equation}\label{oointr}
\Dbar_{n}=\dfrac{1}{t'_n-t'_2}\int_{t'_2}^{t'_n}D(\instant)\,\text{d}\instant.
\end{equation}
Let $Y_i=t'_i-t'_{i-1}$ denote the $i$-th interdeparture time, i.e., the time elapsed between the deliveries of updates $i-1$ and $i$. With $t_{i-1}$ and $t_i$ denoting the arrival/generation times of those same updates, we define  $X_i=t_{i}-t_{i-1}$ as the $i$-th interarrival time, i.e., the time elapsed between
the generation of {\em delivered} updates $i-1$ and $i$.
Partitioning the integral in \eqref{oointr} over the interdeparture intervals $(t'_{i},t'_{i+1})$ yields 
\begin{equation}\label{oointr-partition}
\Dbar_{n}=\dfrac{1}{t'_n-t'_2}\sum_{i=2}^{n-1} \int_{t'_i}^{t'_{i+1}}D(\instant)\,\text{d}\instant.
\end{equation}
As illustrated in Fig.~\ref{Age dispersion}, the integral in \eqref{oointr-partition} equals  the parallelogram area $Q_i=\int_{t'_{i}}^{t'_{i+1}}D(\instant)\,\text{d}\tau$. 
Moreover, for $\tau\in(t'_{i},t'_{i+1})$, $D(\tau)=X_i$ and thus
$Q_i=X_iY_{i+1}$. With the further observation that 
$t'_n-t'_2=\sum_{i=2}^{n-1}Y_{i+1}$, these facts yield
\begin{align}\label{oointr01}
\Dbar_{n}&=\dfrac{\textstyle\sum_{i=2}^{n-1}Q_{i}}{t'_n-t'_2} 
=\dfrac{\textstyle\sum_{i=2}^{n-1}X_iY_{i+1}}{\sum_{i=2}^{n-1}Y_{i+1}}.
\end{align}
%
%
The key idea here  is that the dispersion is $X_i$ over the $(i+1)$ interdeparture interval of length $Y_{i+1}$ and \eqref{oointr01} is averaging the dispersion over these departure intervals.
With the assumption that ${(X_i,Y_i)}$
is a stationary ergodic random process, the average age dispersion is given as
\begin{align}
\Dbar&=\lim_{n\to\infty}\Dbar_{n}
=\dfrac{\mathbb{E}[X_iY_{i+1}]}{\mathbb{E}[Y_{i+1}]}.
\label{oointr02}
\end{align}
Next, we derive the average age dispersion expression for an M/G/1/1 queueing model with probabilistic preemption. 



 \subsection{Age Dispersion: M/G/1/1 with Probabilistic Preemption}\label{mg11 analysis}
 When a packet arrives at an idle server, it is admitted into service. However,  according to the probabilistically
preemptive policy, if a packet arrives  at a busy server, the packet in service is replaced by the arriving packet with the fixed
probability $\theta$ \cite{moltafetisit2025}. 

Fig.~\ref{Age dispersion} shows an example of the AoI process for an M/G/1/1 system. With random variable $T_i$ denoting the system time of the $i$-th \textit{delivered} update in the system, it  can be seen in the figure that the $i$-th interarrival time can be calculated as 
\begin{align}
    X_{i}&=Y_{i}+T_{i-1}-T_{i}.\label{XYT-MG11}
\end{align}
Note that although packet arrivals are Poisson, the interarrival times $X_i$ between delivered packets are generally not exponentially distributed, since preemption selectively filters arriving packets in a manner that depends on both service times and subsequent arrivals.
By substituting $X_{i}$ into \eqref{oointr02}, the average age dispersion is rewritten as
\begin{align}\label{Age_d_general_2}
\Dbar&=\dfrac{\E{Y_iY_{i+1}}+\E{T_{i-1}Y_{i+1}}-\E{T_iY_{i+1}}}{\E{Y_{i+1}}}
    =\E{Y_i},
\end{align}
where the last equality in \eqref{Age_d_general_2} follows because in the M/G/1/1 queue, 
$Y_{i+1}$ is independent of 
$Y_{i}$, and the system times 
$T_{i-1}$ and $T_{i}$, respectively. 

Let $\lambda$ be the packet generation rate and let $S$ denote the exponential $(\mu)$
service time of a packet, with Laplace transform 
$L_S(\theta\lambda)=\E{e^{-\theta\lambda S}}$.
Using Theorem~1 in \cite{moltafetisit2025}, the average interdeparture time in an M/G/1/1 queue with probabilistic preemption is given as 
\begin{align}\label{age_disp_mg11}
   \E{Y}
   =\dfrac{ L_S(\lambda \theta)(\theta- 1)+1}{\lambda \theta L_S(\lambda \theta)},
\end{align}

\begin{remark}\label{Remark-1}
    Using \eqref{age_disp_mg11}, the age dispersions of an M/G/1/1 queue under  the preemptive policy with $\theta =1$  and the non-preemptive policy with  $\theta\rightarrow 0$, shown as $\Dbar_{\text{P}}$ and $\Dbar_{\text{NP}}$, respectively, are given as
\begin{align}\label{Dbar_P}
    &\Dbar_{\text{P}}=
    \dfrac{1}{\lambda L_S(\lambda)},~~~~~~~
    \Dbar_{\text{NP}}=
    \E{S}+\dfrac{1}{\lambda}. 
\end{align}
\end{remark}
A numerical comparison of $\Dbar_{\text{P}}$ and $\Dbar_{\text{NP}}$ is presented in Section~\ref{results}.
Next, we examine  higher-order 
age dispersion.




\section{$k$-th Order Age Dispersion}

Let $D^{(k)}(\instant)$ denote the \emph{$k$-th order age dispersion} at the time instant $\instant$, defined as the difference between the ages of the most recently received update and the $(k+1)$-th most recently received one. That is, $D^{(k)}(\instant)$ is given by
\begin{align}
D^{(k)}(\instant) = t_{N(\instant)} - t_{N(\instant)-k},\qquad \instant\ge t'_{k+1}.
\end{align}
Through the $n$th update delivery at time  $t'_n$, the time average $k$-th order age dispersion at the sink, denoted as $\Dbar^{(k)}_{n}$, is
\begin{align}\label{K-oointr}
\Dbar^{(k)}_{n}&=\dfrac{1}{t'_n-t'_{k+1}}\int_{t'_{k+1}}^{t'_n}D^{(k)}(\instant)\,\text{d}\instant.
\end{align}
With the definition of the $j$th {\em prior age dispersion} 
\begin{align}
    D_j(\instant) = 
    t_{N(\instant)-j+1} - t_{N(\instant)-j},
\end{align}
we observe that the $k$-th order age dispersion can be written as the telescoping sum
\begin{align}\label{Dk-telescoping}
D^{(k)}(\instant) = 
\sum_{j=1}^{k} D_j(\instant). 
\end{align}
It follows from (\ref{K-oointr}) and (\ref{Dk-telescoping}) that  the time average $k$-th order age dispersion  through the $n$th update delivery at the sink 
is
\begin{align}\label{K-oointr-v2}
\Dbar^{(k)}_{n}
&=\dfrac{1}{t'_n-t'_{k+1}}\int_{t'_{k+1}}^{t'_n} \sum_{j=1}^{k} D_j(\instant) 
\,\text{d}\instant.
\end{align}
By partitioning the integral in \eqref{K-oointr-v2} over the interdeparture intervals $(t'_{i},t'_{i+1})$, the  
time-average $k$-th order age dispersion is
\begin{align}
\Dbar^{(k)}_{n}
&=\dfrac{1}{t'_n-t'_{k+1}}\sum_{i=k+1}^{n-1} \sum_{j=1}^{k} \int_{t'_{i}}^{t'_{i+1}}D_j(\instant)
\,\text{d}\instant.\label{K-oointr-sum}
\end{align}
\begin{figure}[t]
\centering
\includegraphics[width=.97\linewidth,trim = 0mm 0mm 0mm 0mm,clip]{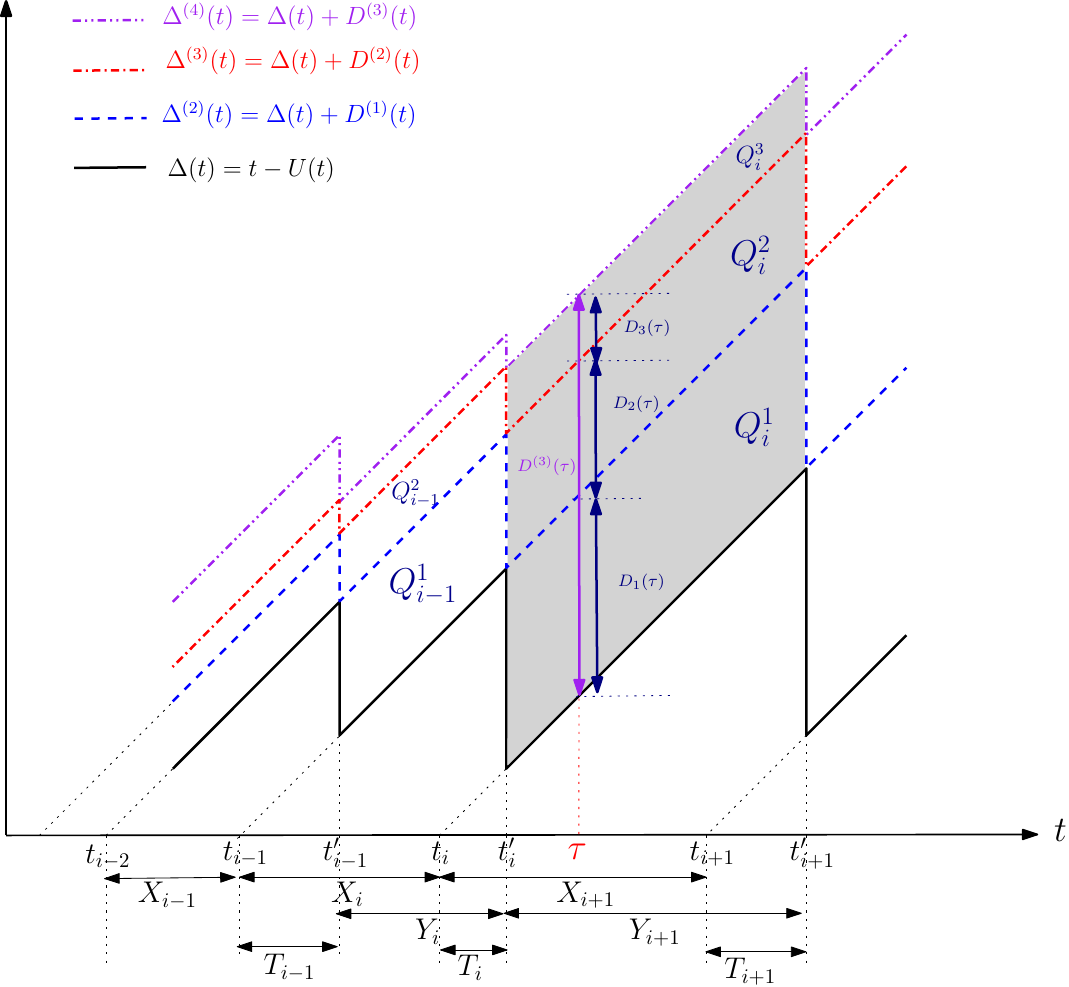}
\caption{An example of the higher-order AoI evolution.}
\label{Age_D_H_O}
\vspace{-3mm}
\end{figure}
As shown in Fig.~\ref{Age_D_H_O}, $D_j(\instant) = X_{i-j+1}$ for $\instant\in (t'_i,t'_{i+1})$  and the parallelogram area denoted $Q_i^j$ satisfies
\begin{align}\label{eqn:Qij-XY}
    Q_i^j= \int_{t'_{i}}^{t'_{i+1}} D_j(\instant)\,\text{d}\instant 
    =Y_{i+1}X_{i-j+1}.
\end{align}
With the observation that $t'_n-t'_{k+1}=\sum_{i=k+1}^{n-1} Y_{i+1}$,  \eqref{K-oointr-sum} and \eqref{eqn:Qij-XY}   imply
\begin{align}\label{K-oointr2}
\Dbar^{(k)}_n
&=\dfrac{\sum_{i=k+1}^{n-1}\sum_{j=1}^{k}Y_{i+1}X_{i-j+1}}{\sum_{i=k+1}^{n-1}Y_{i+1}}.
\end{align}
Just as we observed for the first order dispersion,  \eqref{K-oointr2} reflects that the  $k$th order dispersion is $\sum_{j=1}^{k}X_{i-j+1}$ over the interdeparture interval of length $Y_{i+1}$. 
Under the stationary ergodic assumption of the system, the average $k$-th order age dispersion is given as
\begin{align} 
\Dbar^{(k)}&=\lim_{n\to\infty}\Dbar^{(k)}_{n}
=\dfrac{\mathbb{E}[Y_{i+1}\sum_{j=1}^kX_{i-j+1}]}{\mathbb{E}[Y_{i+1}]}.
\label{k-oointr02}
\end{align}




\begin{definition}
    We define the $k$-th order AoI as the age of the $k$-th most recently received update.
\end{definition}
 Fig.~\ref{Age_D_H_O} illustrates  sample paths of higher-order AoI. In the following remark, we relate the average $k$-th order AoI to the average AoI $\Delta$ and the average $(k-1)$-th order age dispersion.
\begin{lemma}\label{lemma3}
   Let $\Delta^{(k)}$ denote the average $k$-th order AoI. 
   Then, we have 
   \begin{align}
       \Delta^{(k)} &= \Delta + \Dbar^{(k-1)}.
   \end{align}
\end{lemma}

\begin{remark}
   We are interested in studying higher-order AoI because, as shown in Lemma~\ref{lemma3}, a low average $k$-th order AoI implies a small average AoI, yielding fresh information, and a small average $(k-1)$-th order age dispersion, yielding temporally consistent data delivery.
\end{remark}
Next, the average $k$-th order AoI of an  M/G/1/1 system is characterized.


\begin{theorem}
\label{coro_2}
    The average $k$-th order AoI of the M/G/1/1 queueing system with probabilistic preemption is given as
\begin{align}\label{K_AoI_MG}
    \Delta^{(k)}= \Delta+(k-1)\dfrac{ L_S(\lambda \theta)(\theta- 1)+1}{\lambda \theta L_S(\lambda \theta)},
\end{align}
where the AoI $\Delta$ is given as 
\begin{align}
&\Delta\!=\!\dfrac{\!L_{S}(\lambda\theta)((\theta^2\!-\!\theta)(L_{S}(\lambda\theta)\!+\!\!\lambda L'_{S} (\lambda\theta))\!+\!\theta\!-\!1)\!+\!1}{\lambda(\theta^2-\theta) L_{S}(\lambda\theta)^2+\lambda\theta L_{S}(\lambda\theta)}
\end{align}
with $L'_{S} (\lambda\theta)=\mathbb{E}[Se^{-\lambda\theta S}]$. 
\end{theorem}
\begin{proof}
According to Lemma~\ref{lemma3}, to characterize the average $k$-th order AoI, it suffices to derive the average AoI and the average $(k-1)$-th order age dispersion. The average AoI of the M/G/1/1 queueing system was characterized in~\cite[Corollary~1]{moltafetisit2025}. The average $(k-1)$-th order age dispersion of the M/G/1/1 queue is given as
\begin{align}
\Dbar^{(k-1)}&\stackrel{(a)}=\dfrac{\mathbb{E}[Y_{i+1}\sum_{j=1}^{k-1}X_{i-j+1}]}{\mathbb{E}[Y_{i+1}]}\nn&
\stackrel{(b)}=\dfrac{\mathbb{E}[Y_{i+1}]\mathbb{E}[\sum_{j=1}^{k-1}X_{i-j+1}]}{\mathbb{E}[Y_{i+1}]}\nn&
\stackrel{(c)}=(k-1)\mathbb{E}[Y].\label{hiorderdispersion}
\end{align}
In \eqref{hiorderdispersion}, $(a)$ follows from \eqref{k-oointr02}, $(b)$ follows from the fact that in an M/G/1/1 queueing system there is no waiting queue, and hence $Y_{i+1}$ and $X_{i-j+1}$ are independent, and $(c)$ is a consequence of \eqref{XYT-MG11}, which implies $\mathbb{E}[X_{i-j+1}] = \mathbb{E}[Y_{i-j+1}]$ in a stationary M/G/1/1 queue. The claim then follows from \eqref{age_disp_mg11}.
\end{proof}

\section{Numerical Results}\label{results}

\begin{figure}
\centering

\subfigure[$(\kappa, \beta) = (0.5, 0.5)$.]{
\includegraphics[width=0.48\textwidth]{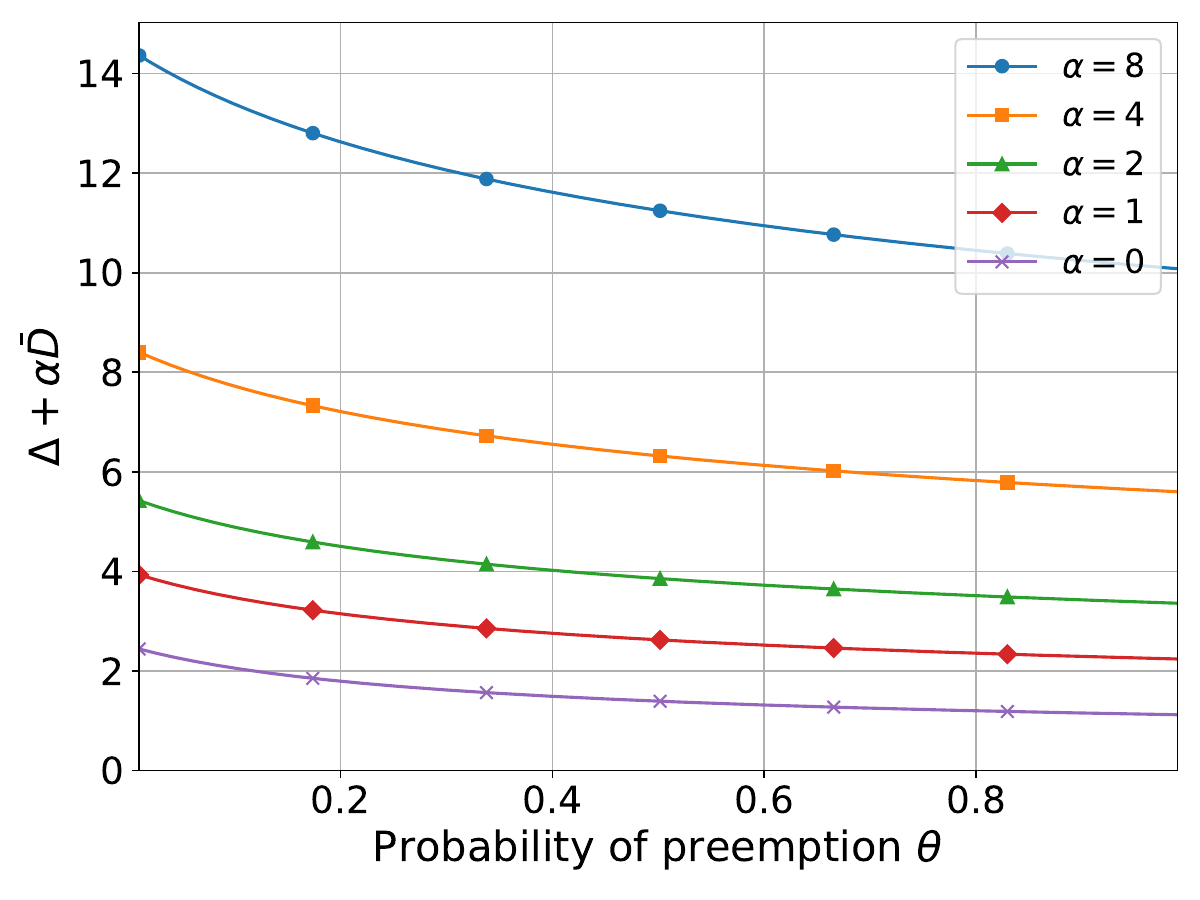}
\label{MG_Gamma_p5}
}
\subfigure[$(\kappa, \beta) = (1, 1)$.]{
\includegraphics[width=0.48\textwidth]{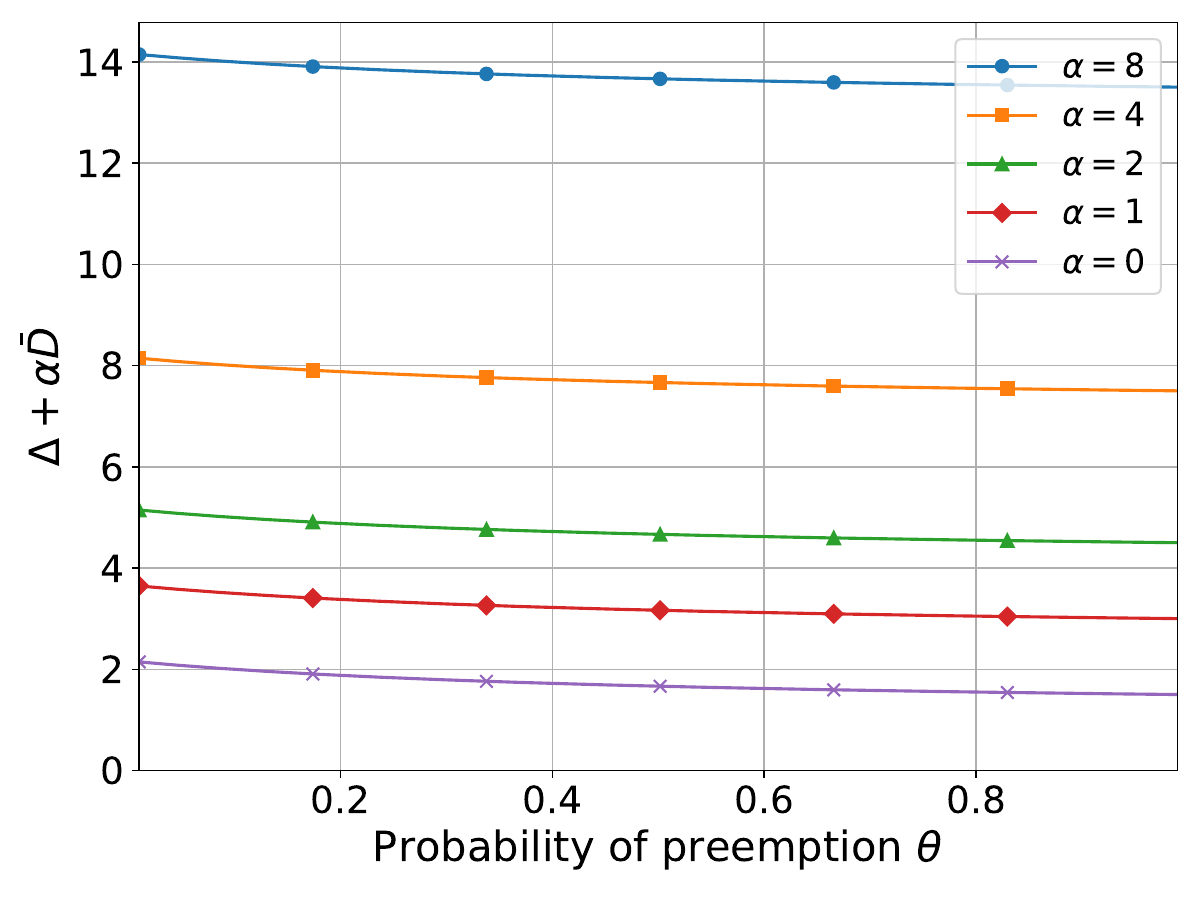}
\label{MG_Gamma_1}
}
\subfigure[$(\kappa, \beta) = (2, 2)$.]
{
\includegraphics[width=0.48\textwidth]{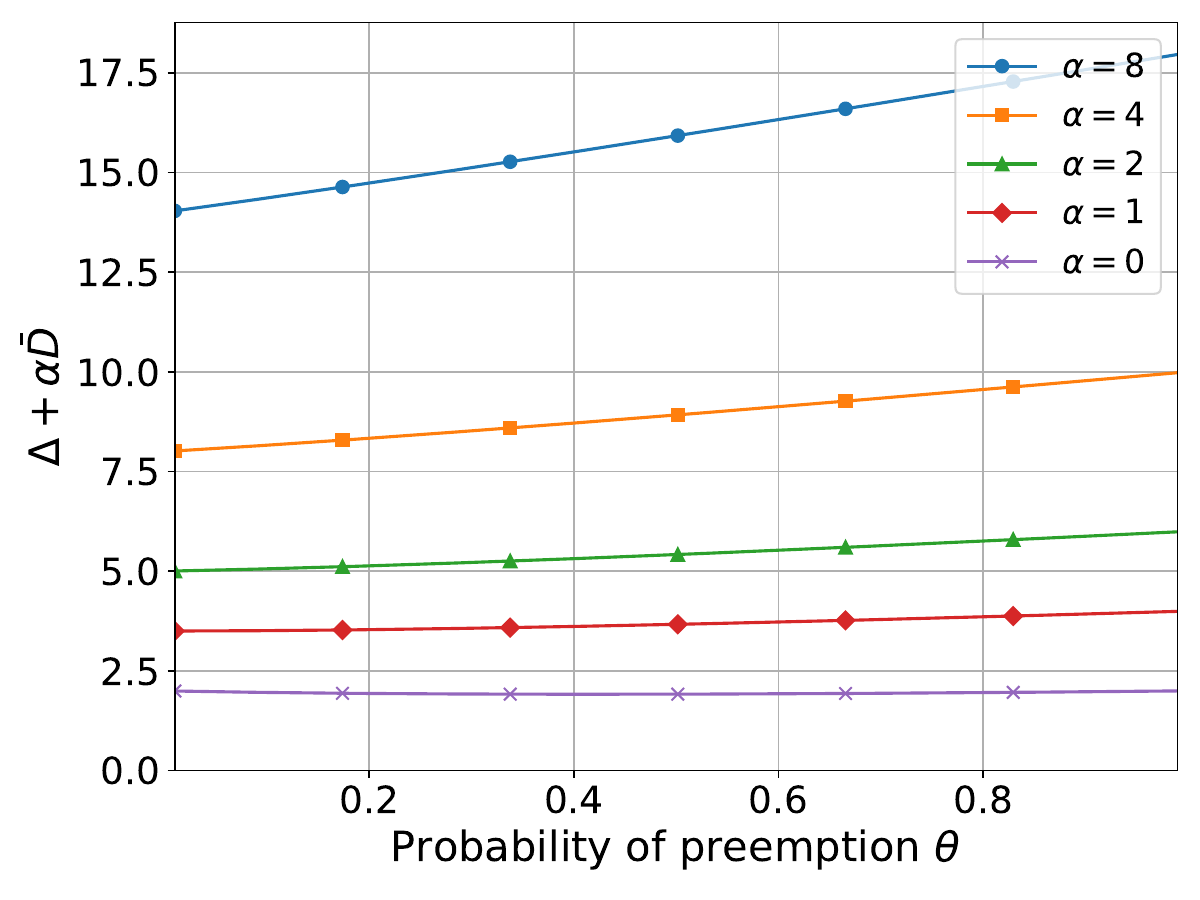}
\label{MG_Gamma_2}
}
\caption{
M/G/1/1 queueing system with probabilistic preemption and Gamma distributed service time: 
$\Delta + \alpha\Dbar$ as a function of $\theta$ for different values of $\alpha$ with $\mu=1$, and $\lambda=2$.}
\label{MG_Gamma}
\end{figure}

\begin{figure}
\centering
\subfigure[Probability of preemption $\theta=1$.]{
\includegraphics[width=0.48\textwidth]{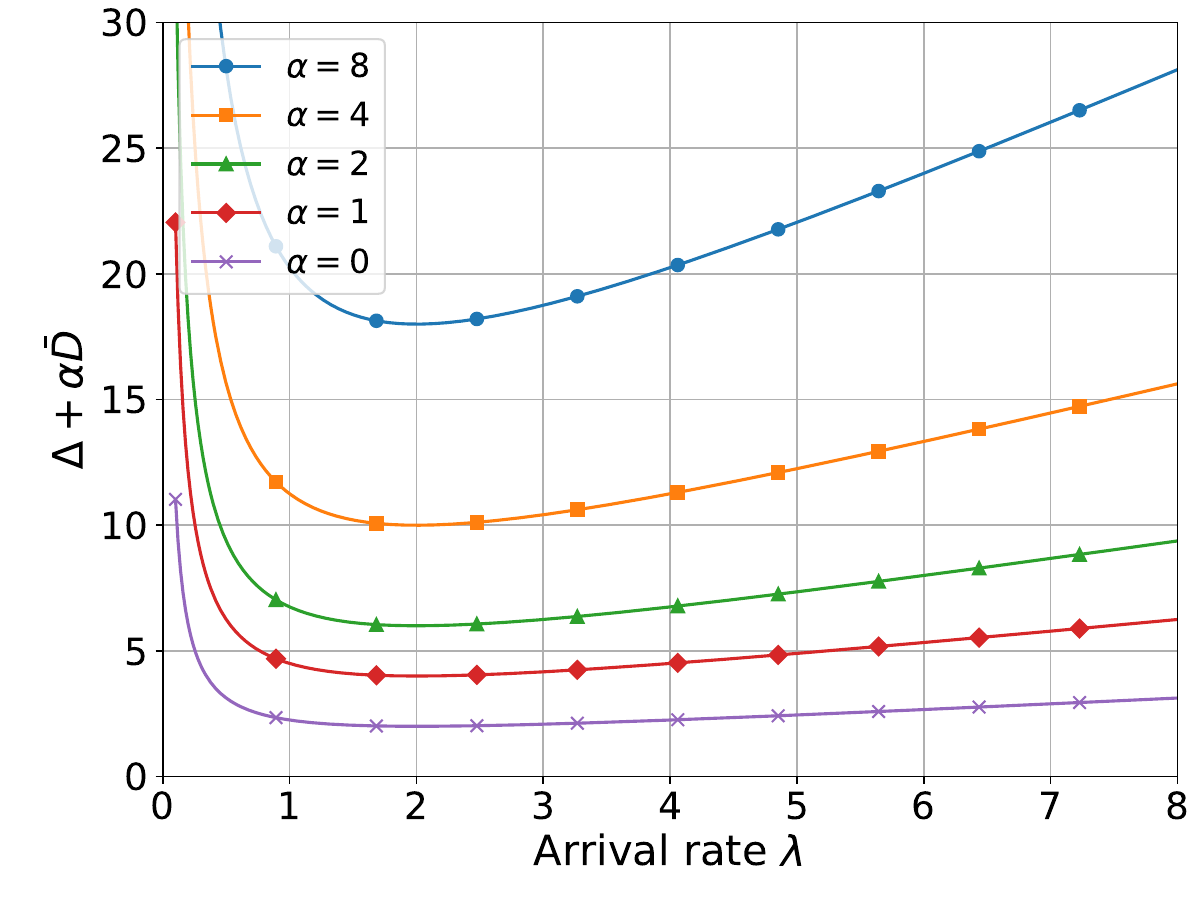}
\label{MG_Gamma_1am_1}
}
\subfigure[Probability of preemption $\theta=0.5$.]{
\includegraphics[width=0.48\textwidth]{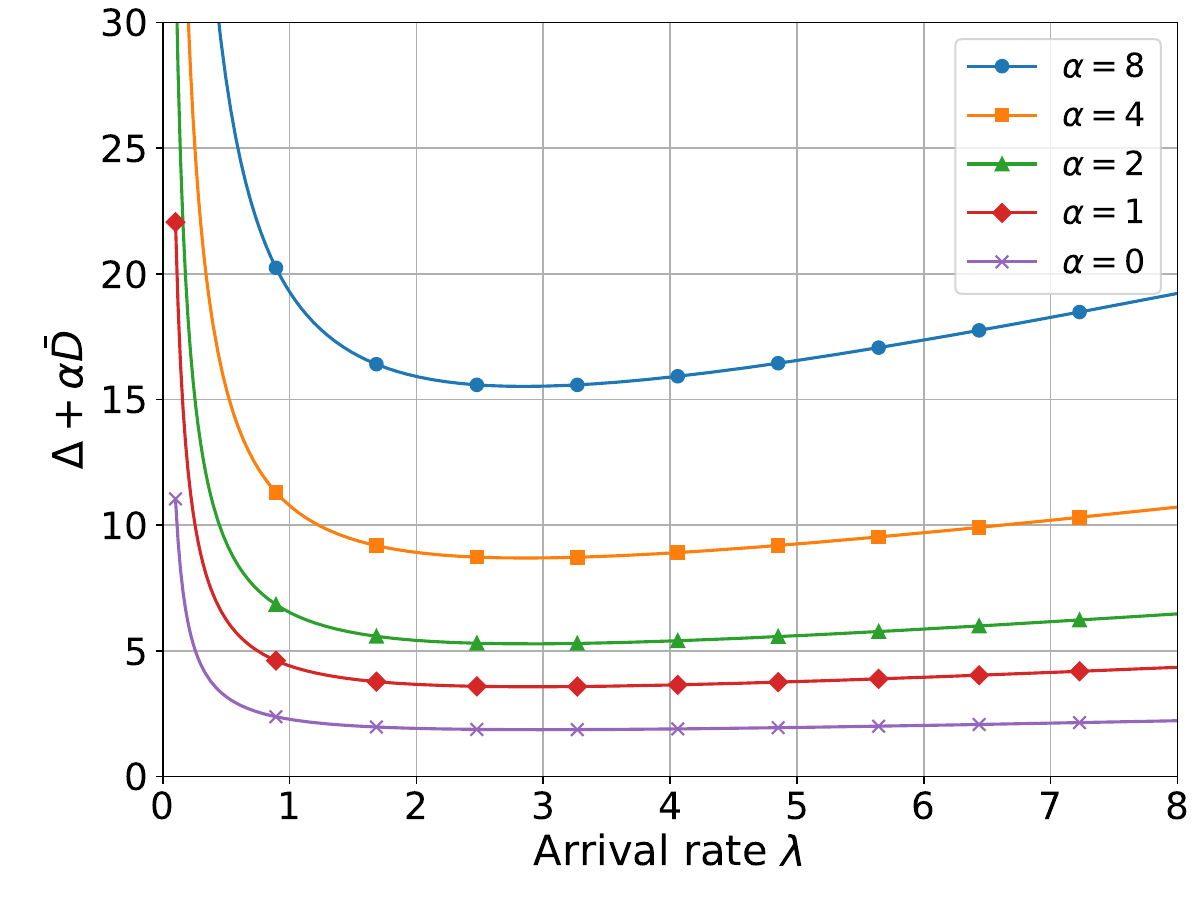}
\label{MG_Gamma_1am_p5}
}
\subfigure[Probability of preemption $\theta=0.001$.]
{
\includegraphics[width=0.48\textwidth]{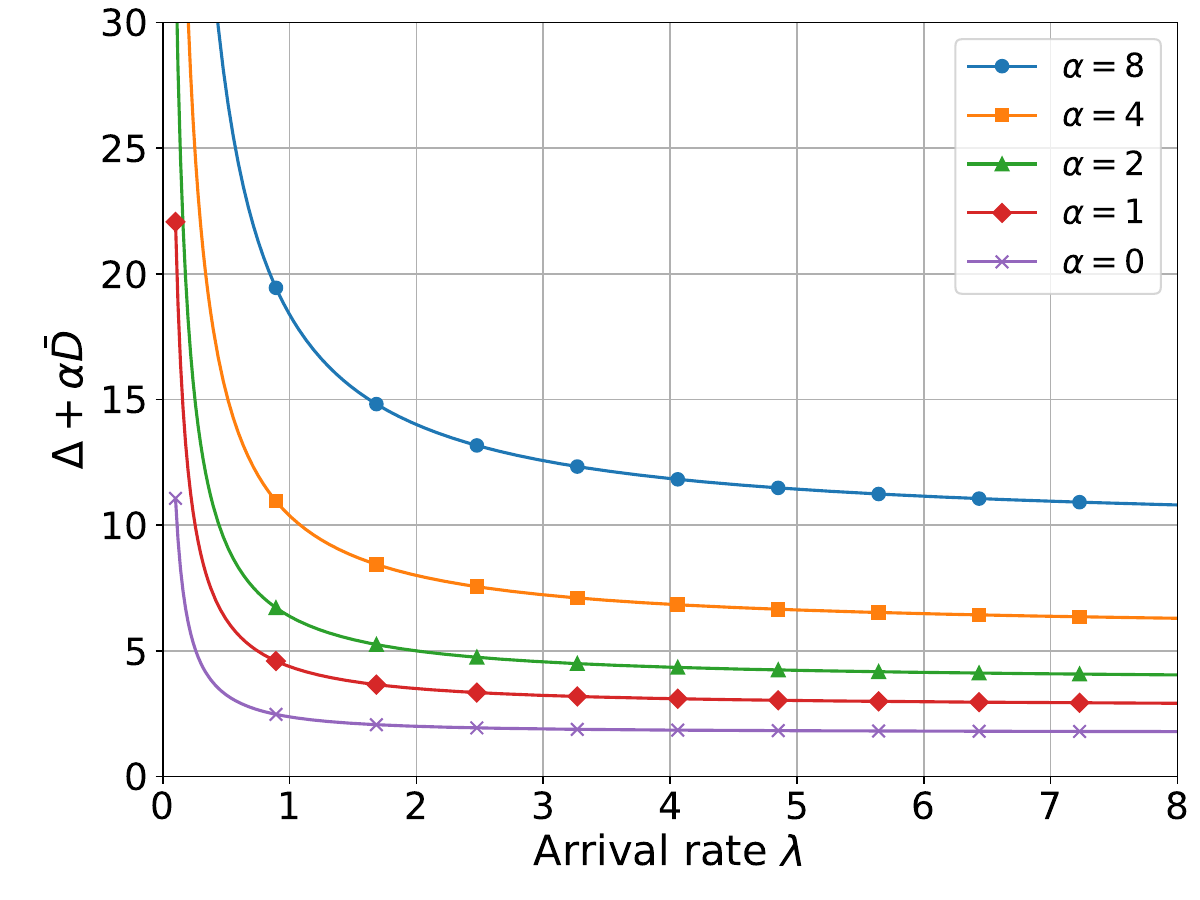}
\label{MG_Gamma_1am_pzz1}
}
\caption{
M/G/1/1 queueing system with probabilistic preemption and Gamma distributed service time:
$\Delta + \alpha\Dbar$ as a function of $\lambda$ for different values of $\alpha$ with $\mu=1$, and $(\kappa, \beta) = (2, 2)$.}
\label{MG_Gamma_1am}
\end{figure}




In this section, we evaluate the weighted sum $\Delta + \alpha \Dbar$ of the average AoI $\Delta$ and the first order age dispersion $\Dbar$ for the M/G/1/1 system under the probabilistically preemptive policy. We keep in mind that this weighted sum is the $k$-th order AoI $\Delta^{(k)}$  when $\alpha=k-1$.

We assume that the service time $S$ follows a Gamma distribution
with parameters shape $\kappa>0$ and rate $\beta>0$. With $\Gamma(\kappa)$ denoting the Gamma function at $\kappa$, the PDF of $S$ is
$f_S(t)=
[\beta^{\kappa}/\Gamma(\kappa)] t^{\kappa-1}\exp(-\beta t),~t>0.$ Note that $S$ has variance ${\kappa}/{\beta}^2
$, and the service rate is $\mu=1/\mathbb{E}[S]={\beta}/{\kappa}$.

 Figs.~\ref{MG_Gamma_p5}, \ref{MG_Gamma_1},  and \ref{MG_Gamma_2} illustrate $\Delta + \alpha\Dbar$ 
 as a function of probability of preemption $\theta$ for different values of $\alpha$ and variance of the distribution.  In these figures, the arrival rate is $\lambda=2$ and the service rate is $\mu=1$. 
 The shape and rate parameters of the Gamma distribution are $(\kappa, \beta) = (0.5, 0.5)$ in Fig.~\ref{MG_Gamma_p5}, $(\kappa, \beta) = (1, 1)$ in Fig.~\ref{MG_Gamma_1}, and $(\kappa, \beta) = (2, 2)$ in Fig.~\ref{MG_Gamma_2}.

 Figs.~\ref{MG_Gamma_1am_1}, \ref{MG_Gamma_1am_p5},  and \ref{MG_Gamma_1am_pzz1} illustrate $\Delta + \alpha\Dbar$ 
 as a function of the arrival rate $\lambda$ for different values of $\alpha$ and probability of preemption $\theta$. In these figures, the shape and rate parameters of the Gamma distribution are $(\kappa, \beta) = (2, 2)$ and the service rate is $\mu=1$. The probability of preemption is $\theta=1$ in Fig.~\ref{MG_Gamma_1am_1},  $\theta=0.5$ in Fig.~\ref{MG_Gamma_1am_p5}, and $\theta=0.001$ in Fig.~\ref{MG_Gamma_1am_pzz1}.

The hazard rate of the service time at time instant $t$ is defined as the instantaneous conditional rate of service completion, given survival of the service process up to time $t$
 \cite{barlow1963properties}. In systems with an increasing hazard rate, an ongoing service is more likely to complete soon, making preemption less advantageous; conversely, when the hazard rate is decreasing or constant, preemption can be more effective. As illustrated in Fig.~\ref{MG_Gamma}, when $\kappa \le 1$ (as in Figs.~\ref{MG_Gamma_p5} and \ref{MG_Gamma_1}), increasing the probability of preemption reduces the average $(\alpha+1)$-th order) AoI, whereas for $\kappa > 1$ (as in Fig.~\ref{MG_Gamma_2}), it increases the average $(\alpha+1)$-th order AoI. This behavior arises because, for the Gamma service time, the hazard rate is decreasing when $\kappa < 1$, constant when $\kappa = 1$, and increasing when $\kappa > 1$~\cite{barlow1963properties}.

Fig.~\ref{MG_Gamma_1am} illustrates that the optimal arrival rate for different orders of the AoI can, in general, vary and depends on the preemption policy adopted in the system. In this figure, the shape parameter is set to $\kappa=2$, for which the hazard rate increases with time. Consequently, preempting an in-service packet is not advantageous. When the probability of preemption is very low, as in Fig.~\ref{MG_Gamma_1am_pzz1}, increasing the arrival rate is beneficial since it leads to more frequent packet deliveries to the sink. However, when preemption is enforced, as in Figs.~\ref{MG_Gamma_1am_1} and \ref{MG_Gamma_1am_p5}, the system must carefully tune the arrival rate. A very high arrival rate leads to the preemption of packets that are likely to be served soon (due to the increasing hazard rate), while a very low arrival rate results in infrequent updates.

To further investigate Remark~\ref{Remark-1}, Fig.~\ref{Dp_Dnp} illustrates $\Dbar_{\text{P}}$ and 
$\Dbar_{\text{NP}}$ as a function of the arrival rate $\lambda$ for different 
values of the shape parameter $\kappa$ of the Gamma distributed service 
time, with $\mu = 1$. Since $\Dbar_{\text{NP}} = \mathbb{E}[S] + 1/\lambda$ 
depends only on the mean service time, and the mean is fixed to 
$\mathbb{E}[S] = 1/\mu$ for all $\kappa$ values, the $\Dbar_{\text{NP}}$ curves 
coincide across all cases. In contrast, $\Dbar_{\text{P}} = \frac{1}{\lambda L_S(\lambda)}$ 
depends on the full service time distribution through the Laplace transform 
$L_S(\lambda)$. As observed in the figure, when $\kappa < 1$ (decreasing 
hazard rate), $\Dbar_{\text{P}} < \Dbar_{\text{NP}}$, indicating that preemption reduces 
age dispersion. When $\kappa = 1$ (exponential service time, constant hazard 
rate), $\Dbar_{\text{P}} =\Dbar_{\text{NP}}$, as the Laplace transform yields 
$L_S(\lambda) = \frac{\mu}{\mu + \lambda}$, making the two expressions 
identical. Finally, when $\kappa > 1$ (increasing hazard rate), 
$\Dbar_{\text{P}} > \Dbar_{\text{NP}}$, since preempting a packet that is likely to 
complete service soon is wasteful and increases dispersion. This behavior 
is consistent with the hazard rate discussion above, where preemption is 
beneficial when $\kappa \leq 1$ and detrimental when $\kappa > 1$. 
Furthermore, $\Dbar_{\text{NP}}$ is decreasing in $\lambda$ for all $\kappa$, 
whereas $\Dbar_{\text{P}}$ is decreasing in $\lambda$ when $\kappa \leq 1$ but 
increasing when $\kappa > 1$, which is consistent with the behavior observed in Fig.~\ref{MG_Gamma_1am}.

\begin{figure}
\centering
\includegraphics[width=.95\linewidth,trim = 0mm 0mm 0mm 0mm,clip]{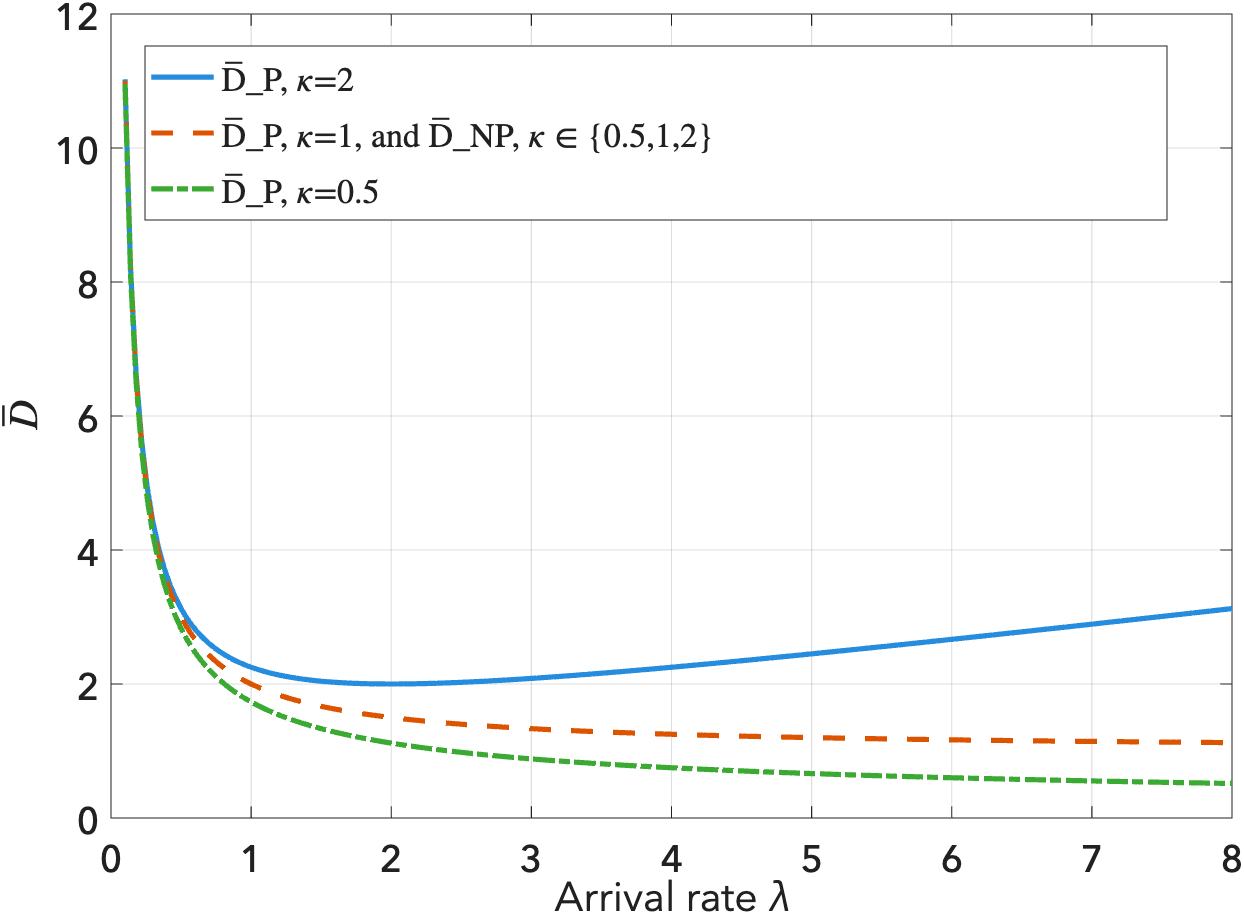}
\caption{Average age dispersion under preemptive ($\Dbar_{\text{P}}$) and non-preemptive 
($\Dbar_{\text{NP}}$) policies as a function of the arrival rate $\lambda$ for Gamma 
distributed service time with different shape parameters $\kappa$ 
and $\mu = 1$.}
\label{Dp_Dnp}
\end{figure}

\section{Conclusions}
We measured temporal consistency in status update systems using the age dispersion. The $k$-th order age dispersion was defined as the difference between the age of the most recently received update and that of the $(k+1)$-th most recently received update. The $k$-th order age dispersion can be used to characterize the $(k+1)$-th order AoI, defined as the age of the $(k+1)$-th most recently received update. We showed that a low value of higher-order AoI implies both fresh and temporally consistent data delivery.

Several interesting directions remain open for future work. While this paper focuses on the M/G/1/1 queueing system, characterizing age dispersion and higher-order AoI in other queueing models, such as M/M/1, M/G/1/2, or multi-source systems, is a natural next step. Another promising direction is the analysis of age dispersion under state-selective preemption policies, where preemption decisions depend on the current age or system state, rather than a fixed probability $\theta$. Finally, extending the framework to generate-at-will systems, where the source can control the timing of update generation, opens the possibility of jointly optimizing age dispersion and AoI through intelligent scheduling.



\bibliographystyle{IEEEtran}
\bibliography{conf_short,jour_short,Bibliography}

\end{document}